\documentclass[10pt]{article}
\usepackage{latexsym,amssymb}
\usepackage{enumerate, verbatim}
\usepackage{graphics}
\usepackage[latin1]{inputenc}
\usepackage[english]{babel}
\usepackage[T1]{fontenc}
\usepackage{amsmath, amsfonts, amsthm, amscd}
\usepackage{amssymb}

\newtheorem{teo}{Theorem}[section]
\newtheorem{prop}{Proposition}[section]
\newtheorem{lemma}{Lemma}[section]
\newtheorem{coro}{Corollary}[section]
\newtheorem{remark}{Remark}[section]

\newcommand{\CC} {\ensuremath{\mathcal{C}}}
\newcommand{\EE} {\ensuremath{\mathcal{E}}}
\newcommand{\VV} {\ensuremath{\mathcal{V}}}
\newcommand{\qq} {\ensuremath{\mathfrak{q}}}
\newcommand{\pp} {\ensuremath{\mathfrak{p}}}
\newcommand{\hh} {\ensuremath{\mathfrak{h}}}

\newcommand{\F} {\ensuremath{\mathbb{F}}}
\newcommand{\wt} {\textnormal{\textrm{wt}}}

\newcommand{\Aut} {\textnormal{\textrm{Aut}}}
\newcommand{\soc} {\textnormal{\textrm{soc}}}
\newcommand{\dne}{\hfill $\Box$ \bigskip}

\begin{document}

\title{The Automorphism Group of an Extremal $[72,36,16]$ Code does not contain elements of order $6$}

\author{ Martino Borello \\
{\small Dipartimento di Matematica e Applicazioni, Universit\`{a} Milano Bicocca}\\
{ \small 20125 Milano, Italy}\\
{\small\tt m.borello1@campus.unimib.it}}
\date{\scriptsize {Version of March 15 2012, printed \today}}
\maketitle

\begin{abstract}
\noindent

The existence of an extremal code of length $72$ is a long-standing
open problem. Let $\CC$ be a putative extremal code of length $72$
and suppose that $\CC$ has an automorphism $g$ of order $6$. We show
that $\CC$, as an $\F_2\langle g\rangle$-module, is the direct sum
of two modules, one easily determinable and the other one which has
a very restrictive structure. We use this fact to do an exhaustive
search and we do not find any code. This prove that the automorphism
group of an extremal code of length $72$ does not contain elements
of order $6$.


\end{abstract}

\section{Introduction}
The existence of an extremal code of length $72$ is a long-standing
open problem \cite{articolo4}. A series of papers investigates the
structure of its automorphism group excluding most of the subgroup
of $\mathcal{S}_{72}$. The most recent result, established in
\cite{articolo12} and \cite{articolo14}, is the following:

\emph{The automorphism group of a binary self-dual doubly-even
$[72,36,16]$ code $\CC$ has order 5 or a divisor of $24$.
Furthermore, if $8$ divides the order of $\Aut(\CC)$ then its Sylow
$2$-subgroup is either $D_8$ or $C_2\times C_2\times C_2$.}

In this paper we will prove that the automorphism group of a binary
self-dual doubly-even $[72,36,16]$ code does not contain elements of
order $6$, obtaining the following.

\begin{teo}
The automorphism group of a binary self-dual doubly-even
$[72,36,16]$ code is either trivial or one of the following: $C_2,
C_3, C_4, C_2\times C_2, C_5, \mathcal{S}_3, D_8, C_2\times
C_2\times C_2, \mathcal{A}_4$ or $\mathcal{S}_4$.
\end{teo}

Notations: $C_n$ is the cyclic group of order $n$; $D_n$ is the
dihedral group of order $n$; $\mathcal{A}_n$ and $\mathcal{S}_n$
are, respectively, the alternating group and the symmetric group of
degree $n$.

With methods similar to those introduced by O'Brien and Willems in
\cite{articolo3}, we strongly use the fact that a binary code with
automorphism group $G$ is an $\F_2G$-module. Consider $g$,
automorphism of order~$6$: we use a variety of results, some new, of
modular representation theory, to study the structure of
$\F_2\langle g\rangle$-modules (for basic concepts of representation
theory see the introduction of \cite{libro2}). In particular we show
that our putative code is the direct sum of two modules: one is the
subcode of words fixed by $g^2$, easily determinable, and the other
one has a socle which belongs, up to equivalence, to a very small
set. From the knowledge of the socle it is quite easy to do an
exhaustive search.
\\ All computations were carried out using \textsc{Magma}
\cite{articolo19}.

\section{Preliminaries}

Let $\CC$ be a binary $[n,k]$ linear code and let $( \ , \ ):\F_2^n
\times \F_2^n \rightarrow \F_2$ be an inner product on $\F_2^n$,
where $\F_2$ is the field with $2$ elements. Then
$$\CC^\perp=\{ \ v \in \F_2^n \ | \ (v,c)=0 \ \ \forall c\in\CC \
\}$$ is the dual of $\CC$, a binary $[n,n-k]$ linear code. $\CC$
terms self-orthogonal if $\CC \subseteq \CC^\perp$ and self-dual if
$\CC=\CC^\perp$. In this paper we are interested in binary linear
codes in $\F_2^n$ with the Euclidean inner product; i.e., if
$u=(u_1,u_2,\ldots,u_n),v=(v_1,v_2,\ldots,v_n)\in\F_2^n$, then
$(u,v)=\sum_{i=0}^n u_i v_i$.

A theorem of Gleason, Pierce and Turyn \cite{articolo7} guarantees
that if $s>1$ divides the weight of every codeword in a nontrivial
binary self-dual code, then $s=2$ or $4$. Binary self-dual codes
automatically satisfy this condition with $s=2$. A binary
doubly-even code is a binary linear code whose words have weight
divisible by $4$. Self-dual doubly-even codes exist only when $n$ is
a multiple of $8$ \cite{articolo8}. A theorem of Mallows and Sloane
\cite{articolo9} shows that for a self-dual doubly-even code
$$d\leq 4 \left[\frac{n}{24}\right]+4$$
where $[x]$ is the integer part of $x$. If $d= 4
\left[\frac{n}{24}\right]+4$, $\CC$ is called extremal. Extremal
codes whose length is a multiple of $24$ are very interesting,
since, for example, all their codewords of a given weight support
five-designs \cite{articolo10}, \cite{articolo11}. The only known
extremal codes of length a multiple of $24$ are the extended binary
Golay code $\mathcal{G}_{24}$ and the extended quadratic residue
code $QR_{48}$, the unique (up to equivalence) extremal codes of
length $24$ and $48$ respectively.

There is a natural (right) action of $\mathcal{S}_n$ on $\F_2^n$
(action on the coordinates): if $v=(v_1,v_2,\ldots,v_n)$ and $g \in
\mathcal{S}_n$, define
$$v^g=(v_{g^{-1}(1)},v_{g^{-1}(2)},\ldots,v_{g^{-1}(n)}).$$ If $\CC$ is a
binary code and $c^g\in\CC$ for all $c\in\CC$, then $g$ is an
automorphism of $\CC$. We denote with $\Aut(\CC)$ the group ($\leq
\mathcal{S}_n$) of the automorphisms of $\CC$.

When $\CC$ is a binary self-dual doubly-even $[72,36,16]$ code, the
possible structure of the elements of $\Aut(\CC)$ is well-known: we
say that an automorphism $h$ is of type $p$-$(c,f)$ if $h$ has $c$
cycles of order $p$ and $f$ fixed points. Then we have the following
result \cite{articolo13}, \cite{articolo15}, \cite{articolo5}.

\begin{prop}\label{cyclestructure}
Let $h$ be an automorphism of prime order of a binary self-dual
doubly-even $[72,36,16]$ code. Then $h$ can be only of the following
types: $2$-$(36,0)$, $3$-$(24,0)$ or $5$-$(14,2)$.
\end{prop}

Let us fix some notations that we will use throughout this paper.

If $v=(v_1,\ldots,v_{n})\in\F_2^{n}$ and
$\Omega=\{j_1,\ldots,j_m\}\subseteq\{1,\ldots,n\}$, we define
$v_{|_{\Omega}}=(v_{j_1},\ldots,v_{j_m})$.

If $\mathcal{W}$ is a subspace of $\F_2^n$ and $h\in\mathcal{S}_n$
define
$$\mathcal{W}(h)=\{w\in\mathcal{W} \ | \
w^h=w\}.$$

If $\Omega_1,\ldots,\Omega_{n_h}$ are the orbits of the action of
$h$ on the coordinates, we say that a $v\in\VV$ is constant on the
orbits of $h$ if $v_{|_{\Omega_i}}$ is either null or the all ones
vector. Obviously, $w\in\mathcal{W}$ belongs to $\mathcal{W}(h)$ if
and only if is constant on the orbits of $h$.

If $\mathcal{W}$ is a subspace of $\F_2^n$, we define
$$\mathcal{W}\otimes\langle(\underbrace{1,\ldots,1}_{m \ \text{times}})\rangle$$
to be the subspace of $\F_2^{n\cdot m}$ obtained extending every
vector of $\mathcal{W}$ substituting every one with $m$ ones and
every zero with $m$ zeros.

\bigskip

\section{Fixed codes}

An usual starting point in the search of a code with certain
automorphisms is the study of the subcode of words fixed by these
automorphisms. This is, in general, an easy problem. In our case it
will provide a fundamental tool to do the exhaustive search.

Let $\CC$ be a self-dual doubly-even $[72,36,16]$ binary code and
suppose $g\in \Aut(\CC)$ such that $\text{o}(g)=6$. By Proposition
\ref{cyclestructure} we have that the elements of order $2$ and $3$
in $\Aut(\CC)$ have no fixed points, so $g$ has no fixed point. Thus
we can suppose, up to a relabelling of the coordinates,
$$g=(1,2,3,4,5,6)\ldots(67,68,69,70,71,72)$$

We have the following known results about $\CC(g^2)$
\cite{articolo12} and $\CC(g^3)$ \cite{articolo14}.

\begin{prop}\label{prop2}
The code $\CC(g^2)$ is equivalent to $\mathcal{G}_{24}\otimes
\langle(1,1,1)\rangle$, with $\mathcal{G}_{24}$ extended binary
Golay code (remember that all the self-dual doubly-even $[24,12,8]$
codes are equivalent and they are all called extended binary Golay
codes).
\end{prop}

\begin{prop}\label{prop3}
The code $\CC(g^3)$ is equivalent to $\mathcal{K}\otimes
\langle(1,1)\rangle$, with $\mathcal{K}$ one of the $41$ self-dual
$[36,18,8]$ codes classified by Mechor and Gaborit
\textnormal{\cite{articolo16}}.
\end{prop}

With argument similar to the ones used in \cite{articolo20} and
\cite{articolo14} we prove the following result about $\CC(g)$.

\begin{prop}\label{prop4}
The code $\CC(g)$ is $\mathcal{H}\otimes \langle (1,1,1,1,1,1)
\rangle$, where $\mathcal{H}$ is equivalent to $\mathcal{F}$, binary
self-dual $[12,6,4]$ code with generator matrix
$$M=\left[\begin{smallmatrix}
1 & 1 & 1 & 1 & 0 & 0 & 0 & 0 & 0 & 0 & 0 & 0 \\
0 & 0 & 1 & 1 & 1 & 1 & 0 & 0 & 0 & 0 & 0 & 0 \\
0 & 0 & 0 & 0 & 1 & 1 & 1 & 1 & 0 & 0 & 0 & 0 \\
0 & 0 & 0 & 0 & 0 & 0 & 1 & 1 & 1 & 1 & 0 & 0 \\
0 & 0 & 0 & 0 & 0 & 0 & 0 & 0 & 1 & 1 & 1 & 1 \\
0 & 1 & 0 & 1 & 0 & 1 & 0 & 1 & 0 & 1 & 0 & 1
\end{smallmatrix}\right].$$
\end{prop}

\begin{proof}
Let $\pi_{12}:\F_2^{72}\rightarrow \F_2^{12}$ the projection such
that
$$(v_1,v_2,v_3,v_4,v_5,v_6,\ldots,v_{67},v_{68},v_{69},v_{70},v_{71},v_{72})\mapsto(v_1,\ldots,v_{67})$$
and $\phi:\F_2^{72}\rightarrow \F_2^{12}$ the map
$$(v_1,v_2,v_3,v_4,v_5,v_6,\ldots,v_{67},v_{68},v_{69},v_{70},v_{71},v_{72})\mapsto\left(\sum_{i=1}^{6}v_i,\ldots,\sum_{i=67}^{72}v_{i}\right).$$
The code $\phi(\CC)$ is self-orthogonal: let $c,c'\in\CC$, then
$$(\phi(c),\phi(c'))=\sum_{j=1}^{12}\left(\sum_{i=6j-5}^{6j}c_i\right)\cdot\left(\sum_{i=6j-5}^{6j}c'_i\right)=\sum_{k=0}^{5}(c^{g^k},c')=0.$$
Furthermore we have $\pi_{12}(\CC(g))=\phi(\CC)^\perp$:\\
let
$f=(f_1,f_1,f_1,f_1,f_1,f_1,\ldots,f_{12},f_{12},f_{12},f_{12},f_{12},f_{12})\in\CC(g)$
and $c\in\CC$. Then
$$(\pi_{12}(f),\phi(c))=\sum_{j=1}^{12}f_{j}\cdot\left(\sum_{i=6j-5}^{6j}c_i\right)=(f,c)=0,$$
so $\pi_{12}(\CC(g))\subseteq \phi(\CC)^\perp$.\\
Viceversa, let $v=(v_1,\ldots,v_{12})\in \phi(\CC)^\perp$. Let
$\overline{v}=(v_1,v_1,v_1,v_1,v_1,v_1,\ldots,v_{12},v_{12},v_{12},v_{12},v_{12},v_{12})$
such that $v=\pi_{12}(\overline{v})$. We have
$$0=(v,\phi(c))=(\overline{v},c)$$
for all $c\in\CC=\CC^\perp$, so $\overline{v}\in\CC$.\\
Thus $v\in\pi_{12}(\CC(g))$ and so
$\pi_{12}(\CC(g))=\phi(\CC)^\perp$.\\
$\pi_{12}(\CC(g))$ is the dual of a self-orthogonal code. It is,
obviously, even and of minimum distance bigger or equal to $4$,
since every word in $\CC(g)$ has weight a multiple of $6$. Thus it
is a $[12,\geq 6,\geq 4]$ even code.

We have just proved that
$\phi(\CC)\subseteq\pi_{12}(\CC(g))=\phi(\CC)^\perp$. Let us suppose
$\phi(\CC)\subsetneq\pi_{12}(\CC(g))=\phi(\CC)^\perp$. So there
exists \mbox{$v\in\phi(\CC)^\perp \setminus \phi(\CC)$}. Denote
$D=\langle \phi(\CC),v\rangle$. This is obviously a self-orthogonal
code. If it is not self-dual we can repeat this algorithm. So we can
find a self-dual code $D'$ such that $\phi(\CC)\subset
D'=D'^\perp\subset \phi(\CC)^\perp$. We have that $D'$ has minimum
distance at least $4$. There is, up to equivalence, only one
self-dual $[12,6]$ code with minimum distance bigger or equal to $4$
\cite{articolo17}, and has generator matrix
$$M=\left[\begin{smallmatrix}
1 & 1 & 1 & 1 & 0 & 0 & 0 & 0 & 0 & 0 & 0 & 0 \\
0 & 0 & 1 & 1 & 1 & 1 & 0 & 0 & 0 & 0 & 0 & 0 \\
0 & 0 & 0 & 0 & 1 & 1 & 1 & 1 & 0 & 0 & 0 & 0 \\
0 & 0 & 0 & 0 & 0 & 0 & 1 & 1 & 1 & 1 & 0 & 0 \\
0 & 0 & 0 & 0 & 0 & 0 & 0 & 0 & 1 & 1 & 1 & 1 \\
0 & 1 & 0 & 1 & 0 & 1 & 0 & 1 & 0 & 1 & 0 & 1
\end{smallmatrix}\right].$$
This code has no overcode of minimum distance bigger or equal to
$4$. So $\phi(\CC)=\pi_{12}(\CC(g))=\phi(\CC)^\perp$ and it is
equivalent to a self-dual code with generator matrix equivalent to
$M$. Setting $\mathcal{H}=\pi_{12}(\CC(g))$ we have the thesis.
\end{proof}

Let us underline the following obvious fact:
\begin{equation}\label{inter} \CC(g)=\CC(g^2)\cap
\CC(g^3).\end{equation} This simply observation will be fundamental
in section \ref{sectclass}.

\bigskip

\section{Decomposition of $\CC$ as $\F_2\langle
g\rangle$-module}\label{dec}

From now on set $\mathcal{V}=\F_2^{72}$, $\CC$ our putative binary
self-dual doubly-even $[72,36,16]$ code and $g$ an automorphism of
$\CC$ of order $6$.

Following \cite{articolo3}, we observe that $\mathcal{V}$ is an
$\F_2\langle g\rangle$-module defining the product in the natural
way:
$$v\cdot\left(\sum_{i=0}^5a_ig^i\right)=\sum_{i=0}^5a_i v^{g^i} \qquad \text{for} \ v\in\VV \ \text{and} \ a_i\in\F_2.$$
The code $\CC$ is an $\F_2\langle g\rangle$-submodule of
$\mathcal{V}$, since $g$ is an automorphism of the code. Obviously
any $\F_2\langle g\rangle$-module is also an $\F_2\langle g^2
\rangle$-module.

In the previous section we considered the fixed codes $\CC(g)$,
$\CC(g^2)$ and $\CC(g^3)$. In particular $\CC(g^2)$ plays an
important role in our method, since $g^2$ has order $3$ and we have
a classical result of Huffman about the decomposition of binary
codes with automorphisms of odd order.
\begin{lemma} (cfr. \textnormal{Lemma $2$ in \cite{articolo18}})
Let $\CC$ be a binary linear code with an automorphism $h$ of odd
order. Then
$$\CC=\CC(h)\oplus \EE(h)$$
where $\EE(h)$ is the subcode of $\CC$ of words of even weight on
the orbits of $h$, i.e.
$$\EE(h)=\{c\in \CC \ | \ \wt(c_{|\Omega_i})\equiv 0 \ (\text{mod} \ 2), \ \forall i\}$$
where $\Omega_i$ are the orbits of $h$ on the coordinates.
\end{lemma}
So, in our case,
$$\CC=\CC(g^2)\oplus\EE(g^2).$$
Let us observe that $\dim\EE(g^2)=\dim\CC-\dim\CC(g^2)=36-12=24$.\\

In order to have more information about $\EE(g^2)$, we need to
reinterpret this decomposition in terms of representation theory. A
basic result in this direction is given by Lemma \ref{Willems}
below. We state it following the notations in \cite{articolo3}.

\begin{lemma}\label{Willems}
Let $\CC$ be a binary linear code, $G\leq\Aut(\CC)$ and
$$1=f_1+\ldots+f_t$$
be a decomposition of $1\in\F_2G$ into \emph{central orthogonal
idempotents} $f_i\in \F_2G$.\\ Set $\VV_i=\VV f_i$ and $\CC_i=\CC
f_i\subseteq \VV_i$ for $i\in\{1,\ldots,t\}$. Then
$$\VV=\VV_1 \oplus \ldots \oplus \VV_t \qquad  \text{and} \qquad \CC=\CC_1\oplus\ldots\oplus\CC_t$$
as $\F_2G$-modules.
\end{lemma}

Let us take $G=\langle g\rangle$. Then
$$f_1=1+g^2+g^4 \qquad \text{and} \qquad f_2=g^2+g^4$$
is a decomposition of $1$ into (central) orthogonal idempotents of
$\F_2\langle g\rangle$.

\begin{remark}\label{rmk1}
\textnormal{ Observing the proof of Lemma $2$ in \cite{articolo18},
it is easy to prove the following facts:
\begin{enumerate}
  \item $\CC_1=\CC f_1=\CC(g^2)$ and $\CC_2=\CC f_2=\EE(g^2)$;
  \item $\VV_1=\VV f_1=\VV(g^2)$, the subspace of all the vectors fixed by
  $g^2$;
  \item $\VV_2=\VV f_2$ is the set of vectors of even weight on the orbits of $g^2$.
\end{enumerate}
They are all $\F_2\langle g\rangle$-modules.}
\end{remark}

\begin{remark}\label{oss}
\textnormal{By direct calculations we deduce the following
properties for the principal ideal
$\mathcal{I}=(f_2)\subset\F_2\langle g\rangle$:
\begin{enumerate}
  \item $\mathcal{I}$ is $4$-dimensional, viewed as an $\F_2$-vector
  space.
  \item $\mathcal{I}$ has only one proper subideal, say
  $\mathcal{J}$. This is $2$-dimensional as $\F_2$-vector space. In particular $\mathcal{J}=\mathcal{I}(1+g^3)$ and $\mathcal{J}$ is the set of elements of
  $\mathcal{I}$ fixed by $g^3$, so that $\mathcal{J}(1+g^3)=0$. If we consider $\mathcal{I}$ as an
  $\F_2\langle g \rangle$-module, then obviously
  it is indecomposable and $\mathcal{J}=\soc(\mathcal{I})$.
  \item As an $\F_2\langle g^2\rangle$-module, $\mathcal{I}$
  has exactly $5$ irreducible $\F_2\langle g^2\rangle$-submodules (all $2$-dimensional),
  including $\mathcal{J}$. Calling
  $\mathcal{L}_1,\mathcal{L}_2,\mathcal{L}_3,\mathcal{L}_4$ the
  others, it holds that $\mathcal{L}_i(1+g^3)=\mathcal{J}$, for all
  $i~\in~\{1,\ldots,4\}$. Obviously they have pairwise trivial
  intersection (as they are irreducibles), so that
  $\mathcal{I}=\mathcal{J}\oplus\mathcal{L}_i$, for all $i\in\{1,\ldots,4\}$.
\end{enumerate}
} \end{remark}

\begin{remark}\label{rmk3} \textnormal{
Let $v\in\VV_2=\VV f_2$ and consider $\mathfrak{m}=v\F_2\langle
g\rangle$, the cyclic $\F_2\langle g\rangle$-module generated
by~$v$. An easy consequence of Remark \ref{oss} is that, if
$v\neq0$, then only two possibilities occur:
\begin{enumerate}
    \item[I.] $\mathfrak{m}\cong\mathcal{J}$, so that $\mathfrak{m}$ is an irreducible $\F_2\langle g \rangle$-module of
    dimension $2$.
    \item[II.] $\mathfrak{m}\cong\mathcal{I}$, so that $\mathfrak{m}$ is an indecomposable $\F_2\langle g \rangle$-module of
    dimension $4$.
\end{enumerate}
In the first case $\mathfrak{m}$ will be called cyclic $\F_2\langle
g\rangle$-module of type I, while in the second case $\mathfrak{m}$
will be called cyclic $\F_2\langle g\rangle$-module of type II. They
inherit by $\mathcal{J}$ and $\mathcal{I}$ all the properties stated
in Remark~\ref{oss}. Let us underline that
$\mathfrak{m}=v\F_2\langle g \rangle$ is a cyclic $\F_2\langle
g\rangle$-module of type I if and only if
$v$ is fixed by $g^3$.\\
Since every $v\in\VV_2$ belongs to a cyclic $\F_2\langle
g\rangle$-submodule of $\VV_2$ (the one generated by itself), all
the irreducible $\F_2\langle g\rangle$-submodules of $\VV_2$ are
cyclic of type I. For the same reason all the irreducible
$\F_2\langle
g^2\rangle$-submodules of $\VV_2$ are of dimension $2$.\\
Every $\F_2\langle g\rangle$-submodule of $\VV_2$ has even
dimension: indeed it is also a $\F_2\langle g^2\rangle$-submodule
and so, by Maschke's Theorem, it is the direct sum of irreducible
$\F_2\langle g^2\rangle$ modules, that have dimension $2$. }
\end{remark}

Now, we give a lemma that relates the socle of $\F_2\langle g
\rangle$-submodules of $\VV_2$ to elements fixed by $g^3$.

\begin{lemma}\label{socleandfixed}
Let $\mathcal{M}$ an $\F_2\langle g \rangle$-submodule of $\VV_2$.
Then
$$\soc(\mathcal{M})=\mathcal{M}(g^3)$$
where $\mathcal{M}(g^3)$ is the set of vectors in $\mathcal{M}$
fixed by $g^3$.
\end{lemma}

\begin{proof}
Since all the irreducible $\F_2\langle g\rangle$-submodules of
$\VV_2$ are cyclic of type I, then all the elements of an
irreducible $\F_2\langle g \rangle$-submodule of $\mathcal{M}$ are
fixed by $g^3$. Viceversa, every element of $\mathcal{M}(g^3)$
belongs necessarily to an irreducible $\F_2\langle g
\rangle$-module, which is therefore contained in $\mathcal{M}$.
Finally, sum of elements fixed by $g^3$ is again fixed by $g^3$.
\end{proof}

\noindent Thus $\soc(\EE(g^2))=(\EE(g^2))(g^3)$. For our purposes it
is important the following lemma, which gives a characterization of
the socle in terms of the subcodes fixed by $g^2$ and $g^3$.

\begin{lemma}\label{Efissati}
With the notations introduced before, it holds
$$\soc(\EE(g^2))=(\EE(g^2))(g^3)=(\CC(g^2)+\CC(g^3))\cap\VV_2.$$
\end{lemma}
\begin{proof}
We have just proved the first equality. Now we prove the second
one.\\
Obviously
$(\EE(g^2))(g^3)=(\EE(g^2))(g^3)\cap\VV_2\subseteq\CC(g^3)\cap\VV_2\subseteq(\CC(g^2)+\CC(g^3))\cap\VV_2$.\\
Viceversa, as we observed in Remark \ref{rmk1}, $\CC(g^2)=\CC f_1$
is an $\F_2\langle g \rangle$-module. Also $\CC(g^3)$ is an
$\F_2\langle g \rangle$-module, since $g$ sends words constant on
the orbits of $g^3$ in words constant on the orbits of $g^3$. So
$\CC(g^2)+\CC(g^3)$ is a $\F_2\langle g \rangle$-module. Thus, by
Lemma \ref{Willems},
$$\CC(g^2)+\CC(g^3)=(\CC(g^2)+\CC(g^3))f_1\oplus(\CC(g^2)+\CC(g^3))f_2.$$
Since
$$\CC(g^2)=\CC(g^2)f_1\subseteq (\CC(g^2)+\CC(g^3))f_1\subseteq \CC f_1=\CC(g^2),$$
we have that
\begin{equation}\label{c2c3scomposizione}
\CC(g^2)+\CC(g^3)=\CC(g^2)\oplus(\CC(g^2)+\CC(g^3))f_2
\end{equation}
and
$(\CC(g^2)+\CC(g^3))f_2=(\CC(g^2)+\CC(g^3))\cap\VV_2\subseteq(\CC\cap\VV_2)=\EE(g^2)$.\\
If $v\in(\CC(g^2)+\CC(g^3))f_2$ then there exist $v_1\in\CC(g^2)$,
$v_2\in\CC(g^3)$ such that $v=(v_1+v_2)f_2$. Obviously $v_1f_2=0$,
so $v=v_2f_2$ is fixed by $g^3$, like $v_2$.
\end{proof}

Let us conclude with a observation that will be crucial in the next
section.

\begin{remark}\textnormal{ We have, by \eqref{inter}, $\dim(\CC(g^2)+\CC(g^3))=12+18-6=24$.\\
So, by \eqref{c2c3scomposizione}, $\dim
((\CC(g^2)+\CC(g^3))\cap\VV_2)= \dim
(\CC(g^2)+\CC(g^3))-\dim(\CC(g^2))=24-12$.\\
Thus, by Lemma \ref{Efissati},
\begin{equation}\label{dimensioni}
\dim(\soc(\EE(g^2)))=12=\frac{\dim(\EE(g^2))}{2}.
\end{equation}
}
\end{remark}

\section{On the structure of the $\F_2\langle g \rangle$-submodules of $\VV_2$}

In this section we will prove a theorem that is a refinement of the
Krull-Schmidt Theorem in a very particular case of $\F_2\langle
g\rangle$-modules. This result gives us a tool to do an exhaustive
search for our code. In this section we strongly use
Remark~\ref{oss} and Remark~\ref{rmk3}.\\
Let us state the theorem.

\begin{teo}\label{decomposition}
Let $\mathcal{M}$ be a $\F_2\langle g \rangle$-submodule of $\VV_2$
such that
$$\dim(\mathcal{M})=2\dim(\soc(\mathcal{M}))=4m.$$
Then for every decomposition
$$\soc(\mathcal{M})=\pp_1\oplus\ldots\oplus\pp_m$$
of the socle in irreducible $\F_2\langle g \rangle$-submodules,
there exist $\qq_1,\ldots,\qq_m$, cyclic $\F_2\langle g
\rangle$-submodules of type~II of $\mathcal{M}$ with
$\soc(\qq_i)=\pp_i$ for all $i\in\{1,\ldots,m\}$, such that
$$\mathcal{M}=\qq_1\oplus\ldots\oplus\qq_m.$$
\end{teo}

Before proving the theorem, we need some lemmas.

\begin{lemma}\label{indsamesocle}
Let $\pp$ an irreducible $\F_2\langle g \rangle$-submodule of type I
of $\VV_2$ and let $\qq_1,\ldots,\qq_n$ be distinct cyclic
$\F_2\langle g \rangle$-submodules of type II of $\VV_2$, all having
$\pp$ as socle, such that
$$\qq_1+\ldots+\qq_n=\pp\oplus\hh_1\oplus\ldots\oplus\hh_n$$
(that is equivalent to ask that $\dim(\qq_1+\ldots+\qq_n)=2+2n$),
where $\qq_i=\pp\oplus\hh_i$ for all $i\in\{1,\ldots,n\}$. Then
$$\dim(\soc(\qq_1+\ldots+\qq_n))=2n.$$
Furthermore, every cyclic $\F_2\langle g \rangle$-submodule of type
II of $\qq_1+\ldots+\qq_n$ has socle $\pp$.
\end{lemma}

\begin{proof}
Let
$$\begin{array}{rrcl} m_{1+g^3}: & \qq_1+\ldots+ \qq_n & \rightarrow & \qq_1+\ldots+\qq_n\\
                                 & v           & \mapsto     &
                                 v(1+g^3)
                                 \end{array}$$
be the linear map of multiplication by $1+g^3$.\\
Since $\hh_1(1+g^3)=\ldots=\hh_n(1+g^3)=\pp$ and $\pp(1+g^3)=0$, we
have
\begin{equation}\label{im}\text{im}(m_{1+g^3})=\pp.
\end{equation}
This implies that $\ker(m_{1+g^3})$ has
dimension $(2+2n)-2=2n$.\\
The element $v\in\qq_1+\ldots+\qq_n$ is fixed by $g^3$ if and only
if $v\in\ker(m_{1+g^3})$, so that, by Lemma \ref{socleandfixed},
$$\soc(\qq_1+\ldots+\qq_n)=\ker(m_{1+g^3}).$$
Furthermore, as for every $\qq$, cyclic $\F_2\langle g
\rangle$-submodule of type II of $\VV_2$,
$\soc(\qq)=m_{1+g^3}(\qq)$, we get the equality $\soc(\qq)=\pp$ by
\eqref{im}.
\end{proof}

\begin{coro}\label{atmost}
Every $\F_2\langle g \rangle$-submodule $\mathcal{M}$ of $\VV_2$
with
$$\dim(\soc(\mathcal{M}))=2m$$ has at most $2^{2m-2}$ cyclic
$\F_2\langle g \rangle$-submodules of type II with the same socle.
\end{coro}

\begin{proof}
Let us fix an irreducible $\F_2\langle g^2\rangle$-submodule $\pp$
of $\soc(\mathcal{M})$ and let
$\mathcal{K}_\pp=\{\qq_1,\ldots,\qq_N\}$ be the set of all cyclic
$\F_2\langle g \rangle$-submodules of type II of $\mathcal{M}$ that
have socle $\pp$. If $N\neq0$, let $\mathcal{N}=\qq_1+\ldots+\qq_N$.
We have that $\dim(\mathcal{N})=2+2n$ for a certain integer $n\geq
1$. Obviously there are $\qq'_1,\ldots,\qq'_n\in \mathcal{K}_\pp$
such that
$$\mathcal{N}=\qq'_1+\ldots+\qq'_n.$$
By Lemma \ref{indsamesocle}, $\dim(\soc(\mathcal{N}))=2n$ and every
cyclic $\F_2\langle g \rangle$-submodule of type II of $\mathcal{N}$
has socle $\pp$ (and so it is contained in $\mathcal{K}_\pp$). Since
every element of $\mathcal{N}\setminus\soc(\mathcal{N})$ is
obviously contained in a cyclic $\F_2\langle g \rangle$-submodule of
type II of $\mathcal{N}$ and $\qq_i\cap\qq_j=\pp$ for $i\neq j$,
$i,j\in\{1,\ldots,N\}$, it is easy to observe that
$$N=\frac{|\mathcal{N}|-|\soc(\mathcal{N})|}{|\qq|-|\soc(\qq)|}=\frac{2^{2n+2}-2^{2n}}{2^4-2^2}=2^{2n-2}.$$
Obviously $\soc(\mathcal{N})\subseteq\soc(\mathcal{M})$. Then the
maximum for $N$ is reached when $\dim(\mathcal{N})=2+2m$.
\end{proof}

\begin{coro}\label{indoverirr}
Every $\F_2\langle g \rangle$-submodule $\mathcal{M}$ of $\VV_2$
such that
$$\dim(\mathcal{M})=2\dim(\soc(\mathcal{M}))=4m$$ has exactly
$2^{2m-2}$ cyclic $\F_2\langle g \rangle$-submodules of type II with
socle $\pp$, for every irreducible $\F_2\langle g\rangle$-submodule
$\pp\subseteq\soc(\mathcal{M})$.
\end{coro}

\begin{proof}
The number of all the irreducible $\F_2\langle g\rangle$-submodules
of $\mathcal{M}$ is equal to
$$N_{\text{irr}}=\frac{|\soc(\mathcal{M})|-1}{2^2-1}=\frac{2^{2m}-1}{2^2-1},$$
since obviously every element of the socle belongs to an irreducible
$\F_2\langle g\rangle$-module and the irreducibles have pairwise
trivial intersection. With similar arguments it is easy to prove
that the number of all the cyclic $\F_2\langle g\rangle$-submodules
of type II of $\mathcal{M}$ is equal to
$$N_{\text{ind}}=\frac{|\mathcal{M}|-|\soc(\mathcal{M})|}{2^4-2^2}=\frac{2^{4m}-2^{2m}}{2^4-2^2}.$$
Then the average of indecomposables for each irreducible is
$$\frac{N_{\text{ind}}}{N_{\text{irr}}}=2^{2m-2},$$
that is also the maximum. So the maximum is reached for every
irreducible $\F_2\langle g\rangle$-submodule.
\end{proof}

\begin{lemma}\label{indoverdiffsocle}
Let $\qq_1,\ldots,\qq_n$ be cyclic $\F_2\langle g
\rangle$-submodules of type II of $\VV_2$ such that
$\pp_1+\ldots+\pp_n$ is a direct sum (i.e.
$\dim(\pp_1\oplus\ldots\oplus\pp_n)=2n$), where $\pp_i=\soc(\qq_i)$
for every $i\in\{1,\ldots,n\}$. Then
\begin{enumerate}
  \item $\qq_1+\ldots+\qq_n$ is a direct sum (i.e. $\dim(\qq_1\oplus\ldots\oplus\qq_n)=4n$),
  \item $\soc(\qq_1+\ldots+\qq_n)=\pp_1\oplus\ldots\oplus\pp_n$.
\end{enumerate}
\end{lemma}
\begin{proof}
For every $i\in\{1,\ldots,n\}$, let $\hh_i$ irreducible $\F_2\langle
g^2\rangle$-submodule such that $\qq_i=\pp_i\oplus\hh_i$.\\
We make induction on $n\geq 2$.\\
Let $n=2$. Obviously $\qq_1\cap\qq_2=\{0\}$. So
$\qq_1+\qq_2=\pp_1\oplus\hh_1\oplus\pp_2\oplus\hh_2$. Arguing as in
Lemma \ref{indsamesocle}, we can consider the linear map $m_{1+g^3}$
and we can prove
$\soc(\pp_1\oplus\hh_1\oplus\pp_2\oplus\hh_2)=\ker(m_{1+g^3})$ has
dimension $4$. Since
$\pp_1\oplus\pp_2\subseteq\soc(\pp_1\oplus\hh_1\oplus\pp_2\oplus\hh_2)$,
the equality holds.\\
Let us suppose \emph{1.} and \emph{2.} true for $n-1$. Then
$\qq_{n}\cap (\qq_1\oplus\ldots\oplus\qq_{n-1})=\{0\},$ since
$\soc(\qq_n)=\pp_n$ has trivial intersection with
$\soc(\qq_1\oplus\ldots\oplus\qq_{n-1})$. Thus \emph{1.} is true for
$n$. \emph{2.} can be proved as in the basis of the induction.
\end{proof}

We have now all the ingredients to prove the theorem.

{\it Proof of Theorem} \ref{decomposition}. Corollary
\ref{indoverirr} implies that the set $\mathcal{K}_{\pp_i}$ of
cyclic $\F_2\langle g \rangle$-submodules of type II of
$\mathcal{M}$ with socle $\pp_i$ is non-empty for all
$i\in\{1,\ldots,m\}$. Choose $\qq_i\in\mathcal{K}_{\pp_i}$ for every
$i\in\{1,\ldots,m\}$. Then, for Lemma~\ref{indoverdiffsocle},
$\qq_1+\ldots+\qq_m$ ($\subseteq \mathcal{M}$) has dimension $4n$.
So the equality holds. \dne

\begin{remark}\textnormal{
For the following we point out that every choice of
$\qq_i\in\mathcal{K}_{\pp_i}$ ($i\in\{1,\ldots,m\}$) is allowed.}
\end{remark}

Coming back to our problem, we observe that $\EE(g^2)$ satisfies, by
\eqref{dimensioni}, the hypothesis of the theorem:
$$24=\dim(\EE(g^2))=2\cdot\dim( \soc(\EE(g^2)))=4\cdot6.$$

In section \ref{sectclass} we will determine the possible socles of
$\EE(g^2)$ and in section \ref{finalalg}, using Theorem
\ref{decomposition} and the results of section \ref{sectclass}, we
will describe the exhaustive search.

\section{Restrictions on $\soc(\EE(g^2))$}\label{sectclass}

By Lemma \ref{Efissati} we know that
$\soc(\EE(g^2))=(\CC(g^2)+\CC(g^3))\cap\VV_2$. This suggests us to
determine which code can be $\CC(g^2)+\CC(g^3)$ in order to get
$\soc(\EE(g^2)$.

We have seen that we can suppose
$g=(1,2,3,4,5,6)\ldots(67,68,69,70,71,72)$, so that
$$g^2=(1,3,5)(2,4,6)\ldots(67,69,71)(68,70,72) \ \ \text{and} \ \
g^3=(1,4)(2,5)(3,6)\ldots(67,70)(68,71)(69,72).$$ Let us recall the
results and the notations of Proposition \ref{prop2}, \ref{prop3}
and \ref{prop4}: $\CC(g^2)$ is equivalent to
$\mathcal{G}_{24}\otimes \langle(1,1,1)\rangle$, with
$\mathcal{G}_{24}$ extended binary Golay code; $\CC(g^3)$ equivalent
to $\mathcal{K}\otimes \langle(1,1)\rangle$, with $\mathcal{K}$ one
of the $41$ self-dual $[36,18,8]$ codes classified by Mechor and
Gaborit; $\CC(g)=\mathcal{H}\otimes \langle (1,1,1,1,1,1) \rangle$,
with $\mathcal{H}$ equivalent to $\mathcal{F}$, self-dual $[12,6,4]$
code with generator matrix $M$.

Define
$$\begin{array}{rcl} \pi_{24}: \CC(g^2)\subset \F_2^{72} & \rightarrow &
\F_2^{24} \\
  (c_1,c_2,c_1,c_2,c_1,c_2,\ldots,c_{23},c_{24},c_{23},c_{24},c_{23},c_{24}) & \mapsto &
 (c_1,c_2,\ldots,c_{23},c_{24})\end{array}$$
and
$$\begin{array}{rcl} \pi_{36}:\CC(g^3)\subset\F_2^{72} & \rightarrow & \F_2^{36} \\
(c_1,c_2,c_3,c_1,c_2,c_3,\ldots,c_{34},c_{35},c_{36},c_{34},c_{35},c_{36})&\mapsto&(c_1,c_2,c_3\ldots,c_{34},c_{35},c_{36})\end{array}$$
Let us call
$$\overline{g}_{24}=(1,2)\ldots(23,24) \qquad \text{and} \qquad \overline{g}_{36}=(1,2,3)\ldots(34,35,36).$$
It is easy to notice that
$$(\pi_{24}(c))^{\overline{g}_{24}}=\pi_{24}(c^g) \qquad \text{and}
\qquad (\pi_{36}(c))^{\overline{g}_{36}}=\pi_{36}(c^g)$$ for $c$ in
$\CC(g^2)$ and $c$ in $\CC(g^3)$ respectively. So
$$\overline{g}_{24}\in \Aut(\pi_{24}(\CC(g^2))) \qquad \text{and} \qquad \overline{g}_{36}\in \Aut(\pi_{36}(\CC(g^3)))$$
This observation implies that $\pi_{36}(\CC(g^3))$ has an
automorphism of order $3$ and degree $36$. Only $13$ out of the $41$
codes classified by Mechor and Gaborit have such an automorphism.

It is easy to see that
$$\pi_{24}(\CC(g))=(\pi_{24}(\CC(g^2)))(\overline{g}_{24}) \qquad \text{and} \qquad \pi_{36}(\CC(g))=(\pi_{36}(\CC(g^3)))(\overline{g}_{36}).$$
Another important observation is that
$$\pi_{24}(\CC(g))=\mathcal{H}\otimes\langle(1,1)\rangle \qquad \text{and} \qquad \pi_{36}(\CC(g))=\mathcal{H}\otimes\langle(1,1,1)\rangle$$

These remarks together with the lemmas below allow us to deduce the
following result.

\begin{teo}\label{c2c3}
$\CC(g^2)+\CC(g^3)$ belongs, up to equivalence, to a set
$\mathcal{L}$ of cardinality $38$. Every element of $\mathcal{L}$ is
a binary self-orthogonal doubly-even $[72,24,16]$ code.
\end{teo}

We describe the algorithm which proves the theorem.

\textbf{Step 1.} Choose a particular extended binary Golay code, say
$\mathcal{G}$, and find all its subcodes equivalent to
$\mathcal{F}\otimes\langle(1,1)\rangle$, say
$\mathcal{M}_1,\ldots,\mathcal{M}_m$.\\
For each $\mathcal{M}_i$ denote with $h_i$ an element of
$\mathcal{S}_{24}$ such that
${\mathcal{M}_i}^{h_i}=\mathcal{F}\otimes\langle(1,1)\rangle$ and
set $\mathcal{G}_i=\mathcal{G}^{h_i}$.\\
Denote $AF=\Aut(\mathcal{F}\otimes\langle(1,1)\rangle)$ and $\mathcal{AG}=\mathcal{G}_1^{AF}$, i.e. the orbit of $\mathcal{G}_1$ under the action of $AF$.\\
Fact: we have $\mathcal{G}_i^{AF}=\mathcal{AG}$ for all
$i\in\{1,\ldots,m\}$ (since, by direct calculations,
$\mathcal{G}_i\in \mathcal{AG}$).

\begin{lemma}
$\mathcal{AG}$ is the set of all extended binary Golay codes that
has $\mathcal{F}\otimes\langle(1,1)\rangle$ as subcode.
\end{lemma}

\begin{proof}
Take an extended binary Golay code $\mathcal{J}$ with
$\mathcal{F}\otimes\langle(1,1)\rangle$ as subcode. Then there
exists $h\in\mathcal{S}_{24}$ such that $\mathcal{J}^h=\mathcal{G}$.
There exists $j\in\{1,\ldots,m\}$ such that
$(\mathcal{F}\otimes\langle(1,1)\rangle)^h=\mathcal{M}_j$.\\ Then
$((\mathcal{F}\otimes\langle(1,1)\rangle)^h)^{h_j}={\mathcal{M}_j}^{h_j}=\mathcal{F}\otimes\langle(1,1)\rangle$
and so $hh_j \in AF$. But $\mathcal{J}^{hh_j}=\mathcal{G}_j$ and so
$\mathcal{J}\in \mathcal{G}_j^{AF}=~\mathcal{AG}$.
\end{proof}

\textbf{Step 2.} Take the $13$ codes of the classification of Mechor
and Gaborit which have automorphisms of order $3$ and degree $36$,
say
$\mathcal{D}_1,\ldots,\mathcal{D}_{13}$.\\
Denote
$$\begin{array}{cl}
\{e_{1,1},\ldots,e_{1,n_1}\} & \subset \Aut(\mathcal{D}_1) ,\\
\vdots & \\
\{e_{13,1},\ldots,e_{13,n_{13}}\} & \subset
\Aut(\mathcal{D}_{13})\end{array}$$ the sets of automorphisms of
order $3$ and degree $36$ of
$\mathcal{D}_1,\ldots,\mathcal{D}_{13}$ respectively.\\
Find $h_{i,j}\in\mathcal{S}_{36}$ such that
$h_{i,j}^{-1}e_{i,j}h_{i,j}=\overline{g}_{36}$ and set
$\mathcal{D}_{i,j}=\mathcal{D}_i^{h_{i,j}}$ (so that
$\overline{g}_{36}$ is an automorphism of
$\mathcal{D}_i^{h_{i,j}}$).\\
Find all the subcodes of $\mathcal{D}_{i,j}$ equivalent to
$\mathcal{F}\otimes\langle(1,1,1)\rangle$ and which are fixed word
by word by $\overline{g}_{36}$, say
$\mathcal{N}_{{(i,j)}_1},\ldots,\mathcal{N}_{{(i,j)}_{n_{i,j}}}$.\\
Denote with $l_{(i,j)_k}$ an element of $\mathcal{S}_{36}$ such that
${\mathcal{N}_{(i,j)_k}}^{l_{(i,j)_k}}=\mathcal{F}\otimes\langle(1,1,1)\rangle$
and that commutes with $\overline{g}_{36}$. Let us show its
existence: every element $l\in\mathcal{S}_{36}$ for which
${\mathcal{N}_{(i,j)_k}}^{l}=\mathcal{F}\otimes\langle(1,1,1)\rangle$
is such that $l^{-1}\overline{g}_{36}l$ has order $3$ and fixes
every word of $\mathcal{F}\otimes\langle(1,1,1)\rangle$; thus, by
direct calculations, it has the same orbits (on the coordinates, as
an element of $\mathcal{S}_{36}$) of $\overline{g}_{36}$. So
$l^{-1}\overline{g}_{36}l=\overline{g}_{36}$ or
$l^{-1}\overline{g}_{36}l=\overline{g}_{36}^{-1}$. In the first case
take $l_{(i,j)_k}=l$, in the second case take $l_{(i,j)_k}=ls$ where
$s$ is an automorphism of $\mathcal{F}\otimes\langle(1,1,1)\rangle$
of order $2$ (a relabelling of the coordinates of each orbit of
$\overline{g}_{36}$) such that
$s^{-1}\overline{g}_{36}^{-1}s=\overline{g}_{36}$.\\ Set
$\mathcal{D}_{(i,j)_k}=\mathcal{D}_{i,j}^{l_{(i,j)_{k}}}$ and denote
with $\CC36$ the set of all the codes
$\mathcal{D}_{(i,j)_k}$.\\
\vspace{-3mm}

$\CC36$ is a proper subset of all the codes equivalent to
$\mathcal{D}_1,\ldots,\mathcal{D}_{13}$ which contain
$\mathcal{F}\otimes\langle(1,1,1)\rangle$. The following lemma shows
that $\CC36$ is big enough to allow us to determine all the possible
$\CC(g^2)+\CC(g^3)$.

\begin{lemma}
There exist $\mathcal{B}_3\in\mathcal{AG}$ and
$\mathcal{B}_2\in\mathcal{C}36$ such that $\CC(g^2)+\CC(g^3)$ is
equivalent to
$${\pi_{24}}^{-1}(\mathcal{B}_3)+{\pi_{36}}^{-1}(\mathcal{B}_2).$$
\end{lemma}

\begin{proof}
Up to equivalence, we can suppose $$\CC(g^2)\cap
\CC(g^3)=\mathcal{F}\otimes\langle(1,1,1,1,1,1)\rangle.$$ There
exists $i\in\{1,\ldots,13\}$ and $h\in\mathcal{S}_{36}$ such that
$\pi_{36}(\CC(g^3))^h=\mathcal{D}_i$.\\
There exists $j\in\{i,\ldots,n_i\}$ such that
$h^{-1}\overline{g}_{36}h=e_{i,j}$. So
$$\pi_{36}(\CC(g^3))^{hh_{i,j}}=\mathcal{D}_{i,j} \quad \text{and}
\quad
\overline{g}_{36}=h_{i,j}^{-1}e_{i,j}h_{i,j}=h_{i,j}^{-1}h^{-1}\overline{g}_{36}hh_{i,j}.$$
There exists $k\in\{1,\ldots,n_{i,j}\}$ such that
$(\mathcal{F}\otimes\langle(1,1,1)\rangle)^{hh_{i,j}}=\mathcal{N}_{(i,j)_k}$.
So
$$\pi_{36}(\CC(g^3))^{hh_{i,j}l_{(i,j)_k}}=\mathcal{D}_{(i,j)_k}
\quad \text{and} \quad
\overline{g}_{36}=l_{(i,j)_k}^{-1}\overline{g}_{36}l_{(i,j)_k}.$$
Thus, if we set $\overline{t}=hh_{i,j}l_{(i,j)_k}$, we have
\begin{enumerate}
  \item $\pi_{36}(\CC(g^3))^{\overline{t}}=\mathcal{D}_{(i,j)_k}$;
  \item
  $\overline{g}_{36}=\overline{t}^{-1}\overline{g}_{36}\overline{t}$;
  \item
  $\overline{t}\in\Aut(\mathcal{F}\otimes\langle(1,1,1)\rangle)$.
\end{enumerate}
It is now possible to construct an element $t\in\mathcal{S}_{72}$
such that
\begin{enumerate}
  \item $\pi_{36}(c^t)=(\pi_{36}(c))^{\overline{t}}$ for all
  $c\in\CC(g^3)$;
  \item $g=t^{-1}gt$;
  \item $t\in \Aut(\mathcal{F}\otimes\langle(1,1,1,1,1,1)\rangle)$.
\end{enumerate}
The construction of $t$ will be done in Remark \ref{wreath}.
Condition $1.$ implies that $\pi_{36}(\CC(g^3)^t)\in\mathcal{C}36$.\\
Condition $2.$ implies that $\CC(g^2)^t$ is equal to $\CC^t(g^2)$.
Actually, if $c\in\CC(g^2)$, then
$(c^t)^{g^2}=(c^t)^{t^{-1}g^2t}=c^t$, so every word of $\CC(g^2)^t$
is fixed by $g^2$. So $\CC(g^2)^t\subseteq\CC^t(g^2)$. $\CC(g^2)^t$
has obviously dimension $12$. $\CC^t$ is a binary self-dual
doubly-even $[72,36,16]$ code with $g$ as automorphism, so
$\CC^t(g^2)$ has dimension $12$ too, and thus it holds the equality.
This implies
that $\pi_{24}(\CC(g^2)^t)$ is an extended binary Golay code.\\
Condition $3.$ implies that $\pi_{24}(\CC(g^2)^t)$ has
$\mathcal{F}\otimes\langle(1,1)\rangle$ as subcode. Indeed,
$$\mathcal{F}\otimes\langle(1,1)\rangle=\pi_{24}(\mathcal{F}\otimes\langle(1,1,1,1,1,1)\rangle)=\pi_{24}((\mathcal{F}\otimes\langle(1,1,1,1,1,1)\rangle)^t)).$$
Thus $\pi_{24}(\CC(g^2)^t)\in~\mathcal{AG}$.
\end{proof}

\begin{remark}\label{wreath}\textnormal{
It is easy to convince themselves that $t$ exists. For reader's
convenience we give an explicit construction of $t$ through wreath
product.\\
Let $\Delta=\{1,2\}$ and $\Gamma=\{1,2,3\}$. We have
$\mathcal{S}_\Delta=\mathcal{S}_2$ and
$\mathcal{S}_\Gamma=\mathcal{S}_3$. We want to explain how the
wreath product $\mathcal{S}_\Delta\wr\mathcal{S}_\Gamma$ acts on the
coordinates of $\F_2^{6}$.\\
Firstly, we can see
$$\Delta\times\Gamma=\Delta_1\cup\Delta_2\cup\Delta_3$$
with $\Delta_1=\{1,4\}$, $\Delta_2=\{2,5\}$ and $\Delta_3=\{3,6\}$.
This can be send in the ordered set $\Omega=\{1,2,3,4,5,6\}$ in a
natural way, that is sending the first element of $\Delta_1$ in the
first element of $\Omega$, the second element of $\Delta_1$ in the
fourth element of $\Omega$, the first element of $\Delta_2$ in the
second element of $\Omega$ and so on, i.e. by sending $i$ in $i$, in
the ordered set $\Omega$. We denote this map
$$\varphi:\Delta_1\cup\Delta_2\cup\Delta_3\rightarrow\Omega.$$
An element $h$ of $\mathcal{S}_\Delta\wr\mathcal{S}_\Gamma$ has the
shape
$$h=(\delta_1,\delta_2,\delta_3,\gamma)\in\mathcal{S}_\Delta\times\mathcal{S}_\Delta\times\mathcal{S}_\Delta\times\mathcal{S}_\Gamma$$
and acts on $\Delta_1\cup\Delta_2\cup\Delta_3$ in the following way:
$$(\Delta_1\cup\Delta_2\cup\Delta_3)^h=(\Delta_{\gamma^{-1}(1)})^{\delta_{\gamma^{-1}(1)}}\cup(\Delta_{\gamma^{-1}(2)})^{\delta_{\gamma^{-1}(2)}}\cup(\Delta_{\gamma^{-1}(3)})^{\delta_{\gamma^{-1}(3)}}.$$
With this notation it is possible to check that, for example,
$$(1,2,3,4,5,6)=\varphi^{-1}(\text{Id},\text{Id},(1,2),(1,2,3))\varphi.$$
In a similar way, we have that
$$\mathcal{S}_{2}\wr\mathcal{S}_{36}=\mathcal{S}_\Delta\wr\mathcal{S}_{\Gamma_{36}},$$
where $\Gamma_{36}=\{1,\ldots,36\}$, acts on the coordinates of
$\F_2^{72}$, thanks to a suitable
$$\varphi_{36}: \underbrace{\{1,4\}}_{\Delta_1} \cup \underbrace{\{2,5\}}_{\Delta_2} \cup \underbrace{\{3,6\}}_{\Delta_3}
 \cup \ldots \cup \underbrace{\{67,70\}}_{\Delta_{34}} \cup \underbrace{\{68,71\}}_{\Delta_{35}} \cup
 \underbrace{\{69,72\}}_{\Delta_{36}}
\rightarrow \Omega_{72}$$ where $\Omega_{72}=\{1,\ldots,72\}$. With
this notation we have that
$$g=\varphi_{36}^{-1}(\text{Id},\text{Id},(1,2),\ldots,\text{Id},\text{Id},(1,2),\overline{g}_{36})\varphi_{36}.$$
Now, the $t\in\mathcal{S}_{72}$ that we were looking for is
$$t=\varphi_{36}^{-1}(\text{Id},\text{Id},(1,2),\ldots,\text{Id},\text{Id},(1,2),\overline{t})\varphi_{36}.$$
This $t$ has all the desired properties (it is checkable by hand). }
\end{remark}

\textbf{Step 3.} Construct the set of all
${\pi_{24}}^{-1}(\mathcal{B}_3)+{\pi_{36}}^{-1}(\mathcal{B}_2)$,
with $\mathcal{B}_3\in\mathcal{AG}$ and
$\mathcal{B}_2\in\mathcal{C}36$ and take one representant for each
equivalence class of this set. Collect them in the set
$\mathcal{L}$, that satisfies, obviously, the requested properties.

\begin{remark}\label{ultima}
\textnormal{ We have proved that, if $\CC$ does exist, there are
$L\in\mathcal{L}$ and $r\in\mathcal{S}_{72}$ such that
$$(\CC(g^2)+\CC(g^3))^r=L.$$
The problem is that it is not clear, a priori, whether $r^{-1}gr=g$
or not. This is very important, since $\VV_2=\VV(g^2+g^4)$ depends
strongly on $g$.\\
The element $r^{-1}gr$ is an automorphism of $L$ of order $6$ and it
has the same cycle structure of $g$. Furthermore
$L=L(r^{-1}g^2r)+L(r^{-1}g^3r)$. There are not many elements with
these features in $\Aut(L)$, $L\in\mathcal{L}$. Using this fact we
construct a new list $\mathcal{L}'$ (of $40$ elements) with the
following property: there exist $L\in\mathcal{L}'$ and
$s\in\mathcal{S}_{72}$ such that
$$(\CC(g^2)+\CC(g^3))^s=L \qquad \text{and} \qquad s^{-1}gs=g.$$
This guarantees that $\soc(\EE(g^2)))^r=L\cap\VV_2$. }
\end{remark}

\section{Construction of a putative $\EE(g^2)$ from its
socle}\label{finalalg}

Now let us conclude the construction. Following the track laid out
in Theorem \ref{decomposition}, we define, for every $\pp$,
irreducible $\F_2\langle g \rangle$-submodule of $\soc(\EE(g^2))$,
the set
$$\mathcal{H}_\pp=\{\hh \ | \ \hh \ \text{irreducible} \ \F_2\langle
g^2 \rangle\text{-submodules of} \ \VV_2 \ \text{such that} \
\hh(1+g^3)=\pp\}.$$ We will now take a decomposition of
$\soc(\EE(g^2))$ into irreducible $\F_2\langle g \rangle$-submodules
$$\soc(\EE(g^2))=\pp_1\oplus\pp_2\oplus\pp_3\oplus\pp_4\oplus\pp_5\oplus\pp_6$$
and add $6$ irreducibles
$\hh_1\in\mathcal{H}_{\pp_1},\hh_2\in\mathcal{H}_{\pp_2},\hh_3\in\mathcal{H}_{\pp_3},\hh_4\in\mathcal{H}_{\pp_4},\hh_5\in\mathcal{H}_{\pp_5},\hh_6\in\mathcal{H}_{\pp_6}$,
to this decomposition, in all the possible combinations. This is
obviously equivalent to take the direct sum of $\qq_1,\ldots,\qq_6$,
cyclic $\F_2\langle g \rangle$-modules of type II such that
$\soc(\qq_i)=\pp_i$ for all $i\in\{1,\ldots,6\}$, in all the
possible combinations.

Now we do some considerations important for the computational
part.\\
Firstly we determine the cardinalities of $\mathcal{H}_\pp$.

\begin{lemma}\label{cardinalities}
For every irreducible $\F_2\langle g \rangle$-submodule $\pp$ we
have $|\mathcal{H}_\pp|=4^{12}$.
\end{lemma}
\begin{proof}
We give a constructive proof. As we have already observed in section
\ref{dec}, every irreducible $\F_2\langle g \rangle$-submodule $\pp$
of $\VV_2$ is cyclic and fixed by $g^3$. Furthermore it is also a
cyclic $\F_2\langle g^2 \rangle$-submodule. Say
$$\pp=\{0,v,v^{g},v^{g^2}\} \qquad \text{with} \ v\in\VV_2(g^3).$$
Consider the sets $\Omega_i=\{6i-5,6i-4,6i-3,6i-2,6i-1,6i\}$, which
we call blocks, for $i\in\{1,\ldots,12\}$. It is easy to observe
that $v_{|_{\Omega_i}}$ can be only one of the following
$$\begin{array}{rlcccrl} \text{\textbf{A.}} & [0,0,0,0,0,0] & & & & \text{\textbf{C.}} & [0,1,1,0,1,1] \\
                         \text{\textbf{B.}} & [1,1,0,1,1,0] & & & & \text{\textbf{D.}} & [1,0,1,1,0,1]
                         \end{array}$$
since $v$ is of even weight on the orbits of $g^2$ and fixed by $g^3$.\\
Every $\hh\in\mathcal{H}_\pp$ is cyclic. Thus we can choose its
generator $z$ so that $z(1+g^3)=v$. Since
\mbox{$z_{|_{\Omega_i}}(1+g^3)=v_{|_{\Omega_i}}$}, the following
possibilities for $z_{|_{\Omega_i}}$ occur:
$$z_{|_{\Omega_i}}\in\left\{[0,0,0,0,0,0], [1,1,0,1,1,0], [0,1,1,0,1,1], [1,0,1,1,0,1] \right\}
\ \text{if} \  v_{|_{\Omega_i}} \ \text{is \textbf{A.}},$$
$$z_{|_{\Omega_i}}\in\left\{[1,0,0,0,1,0], [0,1,0,1,0,0], [1,1,1,0,0,1], [0,0,1,1,1,1] \right\}
\ \text{if} \  v_{|_{\Omega_i}} \ \text{is \textbf{B.}},$$
$$z_{|_{\Omega_i}}\in\left\{[0,1,0,0,0,1], [0,0,1,0,1,0], [1,1,1,1,0,0], [1,0,0,1,1,1] \right\}
\ \text{if} \  v_{|_{\Omega_i}} \ \text{is \textbf{C.}},$$
$$z_{|_{\Omega_i}}\in\left\{[1,0,1,0,0,0], [0,0,0,1,0,1], [0,1,1,1,1,0], [1,1,0,0,1,1] \right\}
\ \text{if} \  v_{|_{\Omega_i}} \ \text{is \textbf{D.}}.$$
We have
$4$ choices for every block and the blocks are $12$. So
$|\mathcal{H}_\pp|=4^{12}$.
\end{proof}

So, apparently, we have ${(4^{12})}^6$ calculations to do, a number
that would make unfeasible our search. The point is that two modules
that ``make $\soc(\EE(g^2))$ grow in the same way'' are equal from
our point of view. More precisely, we are interested only in the
representatives of the equivalence classes of the equivalence
relation over $\mathcal{H}_\pp$ defined as following:
$$\hh_1\sim\hh_2 \qquad \text{if and only if} \qquad \soc(\EE(g^2))+\hh_1=\soc(\EE(g^2))+\hh_2$$

\begin{lemma} Each equivalence class is composed by $4096$ elements.
\end{lemma}
\begin{proof}
Let us fix $\hh\in\mathcal{H}_\pp$, for
$\pp\subseteq\soc(\EE(g^2))$. With arguments similar to the ones
used in Lemma \ref{indsamesocle}, Corollary \ref{atmost} and
Corollary \ref{indoverirr}, we can prove that all the indecomposable
$\F_2\langle g\rangle$-modules in $\soc(\EE(g^2))+\hh$ have socle
$\pp$ and they are $2^{2\cdot6-2}=1024$. Since every indecomposable
$\F_2\langle g\rangle$-module, as we have observed in Section
\ref{dec}, contains $4$ elements of $\mathcal{H}_\pp$, and the
indecomposable $\F_2\langle g\rangle$-modules have pairwise
intersection $\pp$, then there are exactly $4\cdot1024=4096$
elements of $\mathcal{H}_\pp$ in $\soc(\EE(g^2))+\hh$. The thesis
follows easily.
\end{proof}

Thus the number of classes is $4^{12}/4096=4096$, a more practical
number to do calculations. Unfortunately $4096^6$ is still too big.
However, we are interested only in the representatives that give us
a doubly-even module. It is easy to prove (thanks to the
construction in Lemma \ref{cardinalities}) that exactly the half of
the elements of $\mathcal{H}_\pp$ are doubly-even. Moreover, if
$\soc(\EE(g^2))+\hh$ is doubly-even then all the elements of the
class of $\hh$ are doubly-even. It happens that the number of
doubly-even representatives is at most $4^{11}/4096=2048$.

We are to explain our algorithm (we will use the notation
$\text{d}(\CC)$ to indicate the minimum distance of a code $\CC$):

\noindent \textbf{Step 1.} Take the set
$\mathcal{L}'=\{L_1,\ldots,L_{40}\}$ and set $F_i=L_i\cap\VV_2$.

Do the following steps for every $i\in\{1,\ldots,40\}$.

\noindent \textbf{Step 2.} Find $6$ irreducible $\F_2\langle
g\rangle$-modules, say
$$P_{i,1},\ldots,P_{i,6}$$
such that $F_i=P_{i,1}\oplus\ldots\oplus P_{i,6}$.

\noindent \textbf{Step 3.} For every $P_{i,j}$,
$j\in\{1,\ldots,6\}$, find a set of representatives of the $2048$
equivalence classes of $\mathcal{H}_{P_{i,j}}$ consisting of
doubly-even modules. Then find the subset of $\mathcal{H}_{i,j}$
defined in the following way
$$\begin{array}{rl} \mathcal{H}'_{i,j} & =\{H\in\mathcal{H}_{i,j} \ | \ \text{d}(H+L_i)\geq 16\}\\
                                       & =\{H_{{(i,j)}_1},\ldots,H_{{(i,j)}_{n_{i,j}}}\}\end{array}$$

\noindent \textbf{Step 4.} Find the subset of
$\mathcal{H}'_{i,1}\times\mathcal{H}'_{i,2}$ defined in the
following way
$$ \mathcal{P}_i    =\{(H_1,H_2) \ | \ \text{d}(L_i+H_1+H_2)\geq 16\}.$$

\noindent Find the subset of $\mathcal{P}_i\times\mathcal{H}'_{i,3}$
defined in the following way
$$\mathcal{T}_i    =\{(P,H) \ | \ \text{d}(L_i+P_1+P_2+H)\geq 16\}.$$

\noindent Find the subset of $\mathcal{T}_i\times\mathcal{H}'_{i,4}$
defined in the following way
$$ \mathcal{Q}_i    =\{(T,H) \ | \ \text{d}(L_i+T_1+T_2+T_3+H)\geq 16\}.$$

\noindent Find the subset of $\mathcal{Q}_i\times\mathcal{H}'_{i,5}$
defined in the following way
$$\mathcal{F}_i    =\{(Q,H) \ | \ \text{d}(L_i+Q_1+\ldots+Q_4+H)\geq 16\}.$$

\noindent Find the subset of $\mathcal{F}_i\times\mathcal{H}'_{i,6}$
defined in the following way
$$\mathcal{S}_i    =\{(F,H) \ | \ \text{d}(L_i+F_1+\ldots+F_5+H)\geq 16\}.$$

\section{Conclusions}

Theorems \ref{c2c3} tells us that, if a binary self-dual doubly-even
$[72,36,16]$ code with automorphism of order $6$ exists, then it has
a subcode equivalent to one of the $38$ codes in $\mathcal{L}$.
Remark \ref{ultima} and Theorem~\ref{decomposition} imply that the
eventual code can be found in the sets
$$\left\{L_i+S_1+\ldots+S_6 \ | S\in\mathcal{S}_i\right\}_{i\in\{1,\ldots,40\}}$$
where $\mathcal{L}'=\{L_1,\ldots,L_{40}\}$.\\
\textsc{Magma} calculations find $\mathcal{S}_i$ empty, for all
$i\in\{1,\ldots,40\}$. So a binary self-dual doubly-even
$[72,36,16]$ code with automorphism of order $6$ does not exist.

\bigskip

\section*{Acknowledgements}

The author would like to thank his advisors, Prof. Francesca Dalla
Volta and Prof. Massimiliano Sala, for the insightful suggestion of
the problem and for the long and fruitful discussions.\\
\emph{Laboratorio di Matematica Industriale e Crittografia}
(CryptoLab) of \emph{Università degli Studi di Trento} deserves
thanks for the great help in the computational part.

\end{document}